\documentclass[12pt,fleqn,a4paper]{article}

\usepackage{amsmath,amsfonts,amssymb,amsthm,array,graphicx}
\usepackage[english]{babel}
\usepackage{hyperref}
\hypersetup{colorlinks,linkcolor=blue,anchorcolor=blue,citecolor=blue,urlcolor={blue}}
\usepackage{natbib}
\usepackage{enumerate}
\usepackage{doi}
\newtheorem{theorem}{Theorem}

\newtheorem{lemma}[theorem]{Lemma}

\newtheorem{proposition}[theorem]{Proposition}
\theoremstyle{definition}
\newtheorem{definition}[theorem]{Definition}
\newtheorem{example}[theorem]{Example}

\newcommand{\examend}{\hfill \rule[.5ex]{3em}{.15ex} $\diamond$}

\begin{document}

\title{Assignment games with population monotonic allocation schemes\thanks{
The author acknowledges support from the Hungarian National Research, Development and Innovation Office via the grant NKFI K-119930, and thanks participants of the SAET 2021 conference for their comments.}}
\author{Tam\'{a}s Solymosi \thanks{Corvinus Center for Operations Research and Department of Operations Research and Actuarial Sciences, Corvinus University of Budapest. E-mail: \href{mailto:tamas.solymosi@uni-corvinus.hu}{tamas.solymosi@uni-corvinus.hu}}}

\date{\today}
\maketitle
\begin{abstract}
We characterize the assignment games which admit a population monotonic allocation scheme (PMAS) in terms of efficiently verifiable structural properties of the nonnegative matrix that induces the game.
We prove that an assignment game is PMAS-admissible if and only if the positive elements of the underlying nonnegative  matrix form orthogonal submatrices of three special types.
In game theoretic terms it means that an assignment game is PMAS-admissible if and only if it contains a veto player or a dominant veto mixed pair or is composed of from these two types of special assignment games.
We also show that in a PMAS-admissible assignment game all core allocations can be extended to a PMAS, and the nucleolus coincides with the tau-value. 
\vspace{0.15cm}  \vspace{0.15cm}\vspace{0.15cm}
\end{abstract}

\section{Introduction}\label{intro}

In their seminal paper, \cite{ShapleyShubik1972} introduced assignment games to study two-sided matching markets where there are indivisible goods which are traded between sellers and buyers in exchange for money. They proved the non-emptiness of the core and provided a description of the core in terms of the underlying nonnegative pairwise surplus matrix that induces the game via maximal assignment optimization. Thus, the core of an assignment game can be determined without the explicit computation of the (in the number of players exponentially many) coalitional values.
\cite{SolymosiRaghavan2001} presented efficiently verifiable properties of the matrix which characterize various properties of the induced assignment game, for example, stability of the core, largeness of the core, exactness of the game, convexity of the game.

In this paper, following a similar approach, we characterize the assignment games which admit a population monotonic allocation scheme, a concept for core refinement introduced by \cite{Sprumont1990}.
An allocation scheme for a cooperative game specifies how to allocate the worth of every coalition. It is population monotonic if each player's payoff increases as the coalition to which the player belongs to grows larger. \cite{Sprumont1990} showed that, essentially, a game has a population monotonic allocation scheme (PMAS) if and only if it is a positive linear combination of monotonic (simple) games with (individual) veto control. We will see that veto control not only by a single player but also by a seller-buyer pair are the key features of PMAS-admissibility for assignment games.

An obvious necessary condition for a game to admit a PMAS is that the core of each of its subgames is nonempty, that is the game is totally balanced. 
\cite{Sprumont1990} observed that in an (individually) veto controlled monotonic game, we get a PMAS if in each subgame the (or in case of several veto players, a fixed) veto player receives the entire worth and all other players receive nothing. This easy-to-implement rule, however, only leads to a single (extreme) core allocation. We might also be interested which other core allocations can be extended to a full scheme of population monotonic subgame core allocations.
Another sufficient conditions for the existence of a PMAS is convexity of the game. \cite{Sprumont1990} showed that in a convex game all core allocations can be reached by a PMAS. In particular, the Shapley value allocation of the game can be extended to a PMAS if we take the Shapley value allocation in each subgame. 
The PMAS-extendability of all core allocations in convex games seems natural, since convex games exhibit "increasing returns to cooperation" when any player is joining any disjoint coalition.

In this respect, assignment games are quite different from convex games, players on the same side of the market are substitutes rather than complements. No wonder, the class of PMAS-admissible assignment games is ``quite slim'', although, as we will see, larger than the ``very slim'' class of convex assignment games. As \cite{Sprumont1990} demonstrated, if the pairwise surplus matrix contains a $2\times 2$ positive submatrix, no PMAS is possible in the associated assignment game. 
Our goal is to fully explore the structure of the underlying matrix which allows a PMAS in the assignment game. We give our structural characterization in terms of efficiently verifiable matrix properties as well as in terms of decomposibility in ``elementary'' assignment games with veto control (that provides PMAS-admissibility). 
We also show that in an assignment game either none or all core allocations are PMAS-extendable.
Finally, we prove that in PMAS-admissible assignment games, the nucleolus (\cite{Schmeidler1969}) and the tau-value (\cite{Tijs1981}) coincide.

There is a rich literature on subclasses of totally balanced games in which certain extra conditions are sufficient (and maybe also necessary) for admitting a population monotonic allocation scheme.
Without even attempting comprehensiveness, we only mention a few well-known classes of games, for example, minimum cost spanning tree games \citep{NordeMorettiTijs2002}, bankruptcy games \citep{GrahnVoorneveld2002}, simple games \citep{CiftciBormHamers2010}, vertex cover games \citep{XiaoFangDu2020}, weighted multi-sided glove games \citep{MorettiNorde2021}, and independent set games \citep{WangXiaoDuXu2022}.
The paper \citep{AbeLiu2019} discusses monotonic allocation paths in assignment games, a topic cloesely related to population monotonic allocation schemes.
\cite{ElObadiMiquel2019} consider an ``almost'' two-sided variation of assignment games, where there is a special player, called the central player who is a seller and a buyer at the same time, and has veto control.  

Particularly closely related to our work is the recent paper by \cite{XiaoFang2022}. The authors characterize PMAS-admissible matching games defined on edge-weighted undirected networks. Since assignment games are of this type of matching games defined on two-sided networks, the main characterization of \cite{XiaoFang2022} naturally applies to assignment games too. Our paper, developed independently, focuses only on assignment games, but 
also establishes that in assignment games either none or all core allocations are PMAS-extendable, and that in the latter case, the nucleolus and the tau-value coincide.

The rest of the paper is organized as follows.
In the next section, we recall the necessary background and establish two important features of population monotonic allocation schemes, namely that they are independent of inessential coalitions and are decomposable.
In section 3, assignment games with two types of veto control are investigated, namely, when there is a veto player, or when there is a seller-buyer pair with veto control. 
It is proved in section 4 that any PMAS-admissible assignment game is the composition of some of these two ``elementary''  types of assignment games, and that in any PMAS-admissible assignment game all core allocations are PMAS-extendable.  
In section 5, the equivalence of the nucleolus and the tau-value on the class of PMAS-admissible assignment games is proved, and it is illustrated by examples that this coincidence may not hold if either PMAS-admissibility or the assignment nature of the game is lifted. 
Section 6 concludes.

\section{Preliminaries}\label{prelims}

A transferable utility cooperative game on the nonempty finite {\em set of players}, $N$, is defined by a {\em coalitional function} $w:2^N \rightarrow \mathbb{R}$ satisfying $w(\emptyset)=0$.
The function $w$ specifies the worth of every {\em coalition} $S\subseteq N$.
We denote the family of nonempty subsets of players by $\mathcal{P}(N)=2^N\setminus \{\emptyset\}$.
Coalition $S\in\mathcal{P}(N)$ is called \emph{inessential} in game $w$ if $w(S)\leq w(S_1)+\ldots +w(S_k)$ for some of its nontrivial partition $S=S_1\cup\ldots\cup S_k$ consisting of $k\geq 2$ pairwise disjoint nonempty coalitions $S_1,\ldots,S_k\in\mathcal{P}(N)$. Coalitions which are not inessential in $w$ are called \emph{essential} in $w$. We denote the set of essential coalitions in game $w$ by $\mathcal{E}(w)$. Notice that every singleton coalition $\{i\}$, $i\in N$, is essential in any game $(N,w)$.

Given a game $(N,w)$, a {\em payoff allocation} $x\in\mathbb{R}^N$ is called
{\em efficient}, if $x(N)=w(N)$; and
{\em coalitionally rational}, if $x(S) \geq w(S)$ for all $S \subseteq N$;
where, by the standard notation, $x(S)=\sum_{i\in S} x_i$ if $S\neq\emptyset$, and $x(\emptyset)=0$.
We denote by $\mathbf{C}(N,w)$ the {\em core}, that is the set of efficient and coalitionally rational payoffs, of the game $(N,w)$.
Observe that the core is always the same as the \emph{essential-core} where coalitional rationality is required only for the essential coalitions in the game, i.e $\mathbf{C}(N,w)=\mathbf{C}(N,w,\mathcal{E}(w)):=\{x\in\mathbb{R}^N : x(N)=w(N),\; x(S)\geq w(S) \;\forall S \in\mathcal{E}(w)\}$. 
A game is called {\em balanced}, if its core is not empty, and {\em totally balanced}, if all subgames of the game are balanced. 
Convex games and assignment games form two (quite distinct, although not disjoint) subclasses of totally balanced games  such that also all their subgames belong to the same subclass.
Game $(N,v)$ is \emph{convex} if $v(S\cup\{i\})-v(S)\leq v(T\cup\{i\})-v(T)$ holds for all $i\in N$ and $S\subseteq T\subseteq N\setminus\{i\}$.   
Assignment games will be defined in the next section.

The \emph{composition} of games $(N_1,w_1),\ldots,(N_k,w_k)$ with pairwise disjoint player sets is the game $(N,w)$ with player set $N:=N_1\cup\ldots\cup N_k$ and coalitional function $w(S):=w_1(S\cap N_1)+\ldots +w_k(S\cap N_k)$ for all $S\subseteq N$. 
If all games $(N_1,w_1),\ldots,(N_k,w_k)$ are balanced then also the composite game $(N,w)$ is balanced, moreover, 
$x\in\mathbf{C}(N,w)$ if and only if $x=(x^1,\ldots,x^k)$ for some $x^1\in\mathbf{C}(N_1,w_1),\ldots,x^k\in\mathbf{C}(N_k,w_k)$.
A game is said to be \emph{decomposable}, if its player set can be partitioned in at least two nonempty coalitions such that the game is the composition of the corresponding subgames. 
Notice that a coalition which contains players from at least two of the decomposing subgames is inessential in the game.  

\cite{Sprumont1990} introduced the concept of a \textit{population monotonic allocation scheme (PMAS)} as a refinement of core allocations.
\begin{definition}
\label{def:PMAS}
A vector $[x^S_i:S\in\mathcal{P}(N),i\in S]$ is a population monotonic allocation scheme (PMAS) for game $(N,w)$, if the following conditions hold:  
\begin{enumerate}
\item for all $S\in\mathcal{P}(N)$: $\;\sum_{i\in S} x^S_i=w(S)$; \hfill (subgame efficiency)
\item for all $S,T\in\mathcal{P}(N)$ and $i\in S$: $\;S\subseteq T\Rightarrow x^S_i \leq x^T_i$. \hfill (payoff monotonicity)
\end{enumerate}
The subvector $x^S:=[x^S_i:i\in S]$ is called the \emph{suballocation} in PMAS $[x^S_i:S\in\mathcal{P}(N),i\in S]$ corresponding to coalition $S\in\mathcal{P}(N)$.
A payoff vector $x\in\mathbb{R}^N$ in game $(N,w)$ is called \emph{PMAS-extendable} if there exists a PMAS $[x^S_i:S\in\mathcal{P}(N),i\in S]$ for $w$ such that $x=x^N$.
A game $(N,w)$ is called \emph{PMAS-admissible} if it admits a PMAS, that is, if at least one payoff vector is PMAS-extendable in $w$.
\end{definition}

It is clear that only totally balanced games can admit a PMAS, and only core allocations can be PMAS-extendable.
\cite{Sprumont1990} showed that every 3-player superadditive and balanced (hence totally balanced) game is PMAS-admissible, moreover, in these games all core allocations are PMAS-extendable.
Another sufficient property for PMAS-admissibility is convexity of the game. In fact, \cite{Sprumont1990} showed that also in convex games all core allocations are PMAS-extendable.
Thirdly, he observed that in a monotonic veto-controlled game, if we give to a fixed veto player the entire worth of every coalition to which he belongs to and zero to the other coalition members, we get a PMAS. However, in contrast to 3-player totally balanced games and convex games, in a monotonic veto-controlled game the set of PMAS-extendable allocations might be a strict subset of the core, as illustrated in Example~\ref{exam:4player-veto} below. 

On the negative side, existence of a PMAS is not guaranteed for all totally balanced games with at least 4 players.
\cite{Sprumont1990} used assignment games to prove this. 
He showed that ``\emph{every assignment game with at least two sellers and two buyers, and such that each seller-buyer pair derives a positive gain from trade, lacks a PMAS.}''

For reference, we summarize the aformentioned basic results.
\begin{theorem}[\citealp{Sprumont1990}]
\label{theo:Sprumont1990}
\begin{enumerate}
\item A game admits a PMAS if and only if its 0-normalization is the sum of monotonic games with at least one veto player.
\item In any 3-player totally balanced game all core allocations are PMAS-extendable.
\item In any convex game all core allocations are PMAS-extendable.
\item An assignment game admits no PMAS if it is induced by a pairwise surplus matrix which contains a $2\times 2$ positive submatrix.
\end{enumerate}
\end{theorem}
We remark that \cite{Sprumont1990} also described the finitely generated cone of PMAS-admissible games in terms of linear inequalities reminescent of the well-known (Bondareva and Shapley) balancedness conditions for nonemptiness of the core. His dual characterization was substantially strengthened by \cite{NordeReijnierse2002}.
Although streamlining their characterization to assignment games could give further valuable insights to the combinatorial nature of these two-sided matching games, we do not explore this dual approach here. 

We start our discussion with two general observations on population monotonic allocation schemes.
They highlight the inherent structural similarities between PMAS-es and the core.
First, inessential coalitions are redundant for a PMAS, the suballocations corresponding to the essential coalitions determine the entire PMAS.
Secondly, the composition of PMAS-es gives a PMAS for the composite game. 
Since we will frequently employ (but have not found references to) these basic properties of PMAS-es, we also present the proofs.   
\begin{proposition}
\label{prop:inessential-pmas}
\begin{enumerate}
\item Let $[x^S_i:S\in\mathcal{P}(N),i\in S]$ be a PMAS in game $(N,w)$.
If coalition $S\in\mathcal{P}(N)$ is inessential in $w$ because $w(S)\leq w(S_1)+\ldots +w(S_k)$ holds for partition $S=S_1\cup\ldots\cup S_k$ with $k\geq 2$, then $x^{S} = (x^{S_1},\ldots, x^{S_k})$. 
\item If $(N,w)$ is the composition of PMAS-admissible games $(N_1,w_1),\ldots, (N_k,w_k)$  with $k\geq 2$,
then any composition of PMAS-es, one for each game $(N_1,w_1),\ldots, (N_k,w_k)$, gives a PMAS for $(N,w)$, thus, $(N,w)$ is also PMAS-admissible.
\end{enumerate} 
\end{proposition}
\begin{proof}
We prove both statements only for $k=2$, the proofs can be completed by straightforward inductive arguments.

To see Claim 1., let $[x^S_i:S\in\mathcal{P}(N),i\in S]$ be a PMAS in game $(N,w)$, and let $w(S_1\cup S_2)\leq w(S_1)+w(S_2)$ for disjoint nonempty coalitions $S_1,S_2\subseteq N$.
Since $w$ is totally balanced, hence superadditive, we also have $w(S_1\cup S_2)\geq w(S_1)+w(S_2)$, implying $w(S_1\cup S_2)= w(S_1)+w(S_2)$. 
By monotonicity of the individual payoffs, we get $x_{i}^{S_1}\leq x_{i}^{S_1\cup S_2}$ for all $i\in S_1$ and $x_{j}^{S_2}\leq x_{j}^{S_1\cup S_2}$ for all $j\in S_2$. 
Summation of these inequalities and subgame efficiency of suballocations give $w(S_1\cup S_2)= w(S_1)+w(S_2) = x^{S_1}(S_1)+x^{S_2}(S_2)\leq x^{S_1\cup S_2}(S_1\cup S_2)=w(S_1\cup S_2)$.
It follows that all inequalities hold as equations, that is $x_{i}^{S_1}= x_{i}^{S_1\cup S_2}$ for all $i\in S_1$ and $x_{j}^{S_2}= x_{j}^{S_1\cup S_2}$ for all $j\in S_2$. 
Therefore, the $(S_1\cup S_2)$-suballocation $x^{S_1\cup S_2}$ is composed of the $S_1$-suballocation $x^{S_1}$ and the $S_2$-suballocation $x^{S_2}$, that is $x^{S_1\cup S_2}=(x^{S_1},x^{S_2})$.

To see Claim 2., let $(N,w)$ be the composition of PMAS-admissible games $(N_1,w_1)$ and $(N_2,w_2)$. Then any coalition $S\subseteq N$ such that $S\cap N_1\neq\emptyset$ and $S\cap N_2\neq\emptyset$ is inessential in $w$, hence a corresponding suballocation can be obtained as the Cartesian product of an $(S\cap N_1)$-suballocation and an $(S\cap N_2)$-suballocation. It is easily checked that such composition of any PMAS for $(N_1,w_1)$ and any PMAS for $(N_2,w_2)$ gives a PMAS for $(N,w)$.
\end{proof}

The following veto-controlled game demonstrates that not necessarily all core allocations are PMAS-extendable, even if some are.
\begin{example}
\label{exam:4player-veto}
Consider the game with player set $N=\{1,2,3,4\}$ and coalitional function $v(12)=2$, $v(13)=v(123)=3$, $v(14)=v(124)=4$, $v(134)=5$, $v(N)=8$, and $v(T)=0$ for any other coalition $T\subseteq N$.
Since player 1 is a veto player, the game is PMAS-admissible, for example, the core allocation $(8,0,0,0)$ is PMAS-extendable. 
In contrast, allocation $y=(1,2,2,3)$ is in the (essential-)core but it is not PMAS-extendable.
Indeed, only the following simplified scheme (containing only the multi-player essential coalitions, notice that $v(123)=v(2)+v(13)$ and $v(124)=v(2)+v(14)$, so they are also inessential) could extend core allocation $y=(1,2,2,3)$ to a PMAS,
$$
\setlength{\extrarowheight}{1pt}
\begin{array}{|cccc||c|c|cccc|}
\multicolumn{4}{|c||}{S} & v(S) & 0\leq & p^S_1 & p^S_2 & p^S_3 & p^S_4 \\
\hline
\hline
1 & 2 & . & . & 2 & = & p^{12}_1 & p^{12}_2 & . & . \\
1 & . & 3 & . & 3 & = & 1   & .   & 2 & . \\
1 & . & . & 4 & 4 & = & 1   & .   & . & 3 \\
1 & . & 3 & 4 & 5 & = & p^{134}_1 & .   & p^{134}_3 & p^{134}_4 \\
\hline
1 & 2 & 3 & 4 & 8 & = & 1 & 2 & 2 & 3 \\
\hline
\end{array}
$$
because the suballocations corresponding to coalitions $\{1,3\}$ and $\{1,4\}$ (which are tight at the given core payoff vector) could only be efficient if $p^{13}_1=1=y_1$, $p^{13}_3=2=y_3$, and $p^{14}_1=1=y_1$, $p^{14}_4=3=y_4$.
However, by monotonicity of individual payoffs in a PMAS, $p^{134}_1\geq p^{13}_1=1$, $p^{134}_3\geq p^{13}_3=2$, and $p^{134}_4\geq p^{14}_4=3$ must hold, contradicting the subgame efficiency requirement $p^{134}_1+p^{134}_3+p^{134}_4=5$ for coalition $\{1,3,4\}$.
\examend
\end{example}

We have two main goals in this paper: (i) characterize which assignment games 
admit a PMAS, and (ii) answer the question whether all core allocations are PMAS-extendable in PMAS-admissible assignment games.
It follows from \citep{Sprumont1990} (cf. Theorem~\ref{theo:Sprumont1990}) that assignment games with a single seller (or with a single buyer) are PMAS-admissible, because in these (zero-normalized) monotonic games the single seller (or the single buyer) is a veto player.
Secondly, as observed by \cite{SolymosiRaghavan2001}, precisely those assignment games are convex games which can be decomposed in their subgames corresponding to the positive-valued mixed-pair (composed of one seller and one buyer) coalitions.
We will also identify those PMAS-admissible assignment games which are not convex and contain no veto player. 
It will turn out that decomposibility in subgames where a single-player or a mixed-pair has veto power is the key structural feature of PMAS-admissible assignment games.
We will also show that in an assignment game either none or all core allocations are PMAS-extendable.

\section{Assignment games}\label{assignmentgames}

Given two disjoint finite sets $S$ and $T$, we call $\mu\subseteq S\times T$ an {\em $(S,T)$-assignment}, if it is a bijection from some $S'\subseteq S$ to some $T'\subseteq T$ such that $\vert S'\vert=\vert T'\vert=\min(\vert S\vert,\vert T\vert)$.
Trivially, $\mu=\emptyset$, if $S=\emptyset$ or $T=\emptyset$.
We shall write $(i,j)\in\mu$ as well as $\mu(i)=j$ and $\mu(j)=i$.
We denote by $\mathcal{M}(S,T)$ the set of all $(S,T)$-assignments.
Obviously, $\mathcal{M}(S,T)=\{\emptyset\}$, if $S=\emptyset$ or $T=\emptyset$.

A game $(N,w)$ is called an {\em assignment game}, if there exists a partition $N=I\cup J$, $I\cap J=\emptyset$, of the player set and a nonnegative matrix $A=[a_{ij}]_{i\in I, j\in J}$ such that
$$
w(S)=w_A(S):=\max_{\mu\in\mathcal{M}{(S\cap I, S\cap J)}} \;
     \sum_{(i,j)\in \mu} a_{ij}   \qquad \mbox{ for all } S\subseteq N.
$$
A matching $\mu\in\mathcal{M}(S\cap I, S\cap J)$ such that $w_A(S)=\sum_{(i,j)\in\mu}a_{ij}$ is an {\em optimal matching} for coalition $S\subseteq N$. 
Due to the non-negativity of $A$, we can (and will) assume (without loss of generality) that any optimal matching is a complete matching for the `short side' of the coalition.
We denote by $\mathcal{M}_A^*(S\cap I,S\cap J)$ the set of optimal matchings for coalition $S$. 
A player in $I$ or $J$ is called a row or column player, respectively.
Coalitions containing one row player and one column player are called mixed-pair coalitions and will be also be written as ordered pairs of the two players, always the row player written first.
More formally, $(i,j)$ denotes the mixed-pair coalition $\{i,j\}$, $i\in I$, $j\in J$.
Notice that $w_A(\{i,j\})=a_{ij}$ for all mixed-pair coalitions $(i,j)$, and $w_A(\{k\})=0$ for all players $k\in I\cup J$, thus, only the single-player and the positive-valued mixed-pair coalitions are essential in assignment game $(I\cup J,w_A)$.
Formally, $\mathcal{E}(w_A)=\big\{ \{k\} : k\in I\cup J\big\} \cup \big\{ \{i,j\} : a_{ij}>0, i\in I, j\in J\big\}$.

It is easily checked that assignment games are zero-normalized and superadditive, hence nonnegative and monotonic.
\cite{ShapleyShubik1972} proved that every assignment game $(I\cup J,w_A)$ is totally balanced, and the core is the set of solutions to quadratic many linear constraints given in terms of the underlying matrix $A$.
Namely, with any fixed optimal matching $\mu^*\in\mathcal{M}_A^*(I,J)$,
$$ 
\mathbf{C}(I\cup J,w_A)=\left\{(u,v)\in\mathbb{R}^I\times\mathbb{R}^J: 
\begin{array}{c@{}c@{}c@{}c@{}cl}
u_i & + & v_j & = & a_{ij} & (i,j)\in\mu^* \\
u_i &   &     & = & 0      & i\in I\setminus\mu^*(J) \\
    &   & v_j & = & 0      & j\in J\setminus\mu^*(I) \\
u_i & + & v_j & \geq & a_{ij} & (i,j)\in I\times J\setminus\mu^* \\
u_i & , & v_j & \geq & 0      & i\in\mu^*(J), j\in\mu^*(I) \\
\end{array} 
\right\},
$$ 
where $I\setminus\mu^*(J)$ and $J\setminus\mu^*(I)$ denote, respectively, the set of row and column players unmatched under $\mu^*$. 

\cite{ShapleyShubik1972} also proved that the core of any assignment game has a \emph{row-optimal} extreme point where each row player receives his core maximum payoff and each column player receives her core minimum payoff. Similarly, the core also has a \emph{column-optimal} extreme point where each column player receives her core maximum payoff and each row player receives his core minimum payoff.
We denote the row-optimal and the column-optimal core allocations by $(\overline{u},\, \underline{v})$ and $(\underline{u},\, \overline{v})$, respectively, where $\overline{u}_i=\max \{u_i: (u,v)\in\mathbf{C}(w_A)\}$ and $\underline{u}_i=\min \{u_i: (u,v)\in\mathbf{C}(w_A)\}$ for all $i\in I$, similarly, $\overline{v}_j=\max \{v_j: (u,v)\in\mathbf{C}(w_A)\}$ and $\underline{v}_j=\min \{v_j: (u,v)\in\mathbf{C}(w_A)\}$ for all $j\in J$.

The average of the two side-optimal core allocations was recommended by \cite{Thompson1981} as the ``fair'' allocation favoring neither side in the core. It was shown to coincide with the tau-value (the definition is presented later)
of the assignment game by \cite{NunezRafels2002}.
We remark that in assignment games the tau-value is ``typically'' different from the nucleolus (see e.g. Example~\ref{exam:4player-assign-tau}), let alone from the Shapley value (that is ``typically'' not even in the core).
We will show (in Theorem~\ref{theo:tau=nucleolus}), however, that in PMAS-admissible assignment games the tau-value and the nucleolus coincide, yet they could be different from the Shapley value (even if that allocation happens to be in the core).

Since in convex games the Shapley value is always a core allocation and the array of Shapley values of the subgames form a PMAS \citep{Sprumont1990}, for reference, we recall the following characterization of convex assignment games.
\begin{proposition}[\citealp{SolymosiRaghavan2001}]
\label{prop:assignment-convex}
An assignment game is convex if and only if it is induced by a (column permuted) diagonal matrix, or equivalently, if it is the composition of $(1+1)$-player assignment games.
\end{proposition}
The Shapley value, the nucleolus, and the tau-value are well known to coincide in 2-player superadditive (or equivalently, convex or assignment) games.
Combined with the decomposability of the Shapley value and the nucleolus on the class of totally balanced games,  Proposition~\ref{prop:assignment-convex} and Theorem~\ref{theo:tau=nucleolus} (where we prove that the nucleolus and the tau-value are the same in PMAS-admissible, in particular in convex games) imply that all three aforementioned point-valued solutions coincide in convex assignment games.
 
In the following subsections we discuss special types of assignment games which will turn out to be the key building blocks for our main characterization result (Theorem~\ref{theo:assign-PMAS-char}).

\subsection{Assignment games with a veto player}

Suppose assignment game $(I\cup J, w_A)$ is free of null players (i.e., there is no player $i$ such that $w_A(S\cup\{i\})=w_A(S)$ for all $S\subseteq N\setminus\{i\}$), or equivalently, that the underlying matrix $A$ contains neither a full zero row nor a full zero column. Then there is a veto player in $w_A$ if and only if $\vert I\vert=1$ or $\vert J\vert=1$ or both.
In other words, row player $i\in I$ is a veto player in assignment game $(I\cup J, w_A)$ if and only if all positive entries of $A$ are in row $i$. In this case all row players in $I\setminus\{i\}$ (if any) are null players. Similarly for column player $j\in J$.
It is then clear that there is a veto player in assignment game $(I\cup J, w_A)$ if and only if only one row or only one column of $A$ contains positive entries.

We know from \citep{Sprumont1990} that the ``veto player gets everything and all other players get nothing'' allocation rule generates a PMAS in any monotonic game with a veto player (a veto player is obviously also a veto player in any subgame).
The natural question arises whether this extreme allocation is the only PMAS-extendable core allocation in veto-controlled monotonic games? It is clear that any other player in any subgame not containing a veto player is a null player, hence should get zero in any core allocation in the subgame. Could this be otherwise in subgames containing the veto player?

Next we show that if a veto-controlled monotonic game is an assignment game then all core allocations are PMAS-extendable.
\begin{proposition}
\label{prop:assign-veto}
If an assignment game has a single row player or a single column player then every core allocation is PMAS-extendable.
\end{proposition}
\begin{proof}
We show the case when the assignment game is induced by an $1\times \vert J\vert$ nonnegative matrix. The case when there is a single column player is proved analogously.

If $\vert J\vert\leq 2$, i.e. there are at most three players, the claim follows from Theorem~\ref{theo:Sprumont1990}. 

Let $\vert J\vert=3$, $I=\{1\}$, $J=\{2,3,4\}$, and the assignment game $w$ be induced by the single row matrix $[a, b, c]$ where we assume without loss of generality that $a \geq b\geq c$.
Then $w(I\cup J)=a=w(\{1,2\})$.
Since players 3 and 4 are optimally not matched, they get 0 payoff in any core allocation. 
Thus the core is the set of allocations $\{(x_1,x_2,0,0) : x_1+x_2=a, x_1\geq b, x_2\geq 0\}$.
It follows from Proposition~\ref{prop:inessential-pmas} that the following simplified scheme (containing only the multi-player essential coalitions and the grand coalition) determines a PMAS which extends any arbitrarily chosen core allocation $(x_1,x_2,0,0)$:
$$
\setlength{\extrarowheight}{1pt}
\begin{array}{|c|ccc||c|c|c|ccc|}
\multicolumn{4}{|c||}{S} & w(S) & 0\leq & p^S_1 & p^S_2 & p^S_3 & p^S_4 \\
\hline
\hline
1 & 2 & . & . & a & = & x_1 & x_2 & . & . \\
1 & . & 3 & . & b & = & b   & .   & 0 & . \\
1 & . & . & 4 & c & = & c   & .   & . & 0 \\
\hline
1 & 2 & 3 & 4 & a & = & x_1 & x_2 & 0 & 0 \\
\hline
\end{array}
$$

In case of $\vert J\vert\geq 4$, any column player with pairwise surplus not among the highest three surplusses can be inserted in the PMAS with 0 payoffs everywhere when involved, just like player 4 in the above scheme.  
\end{proof}

\subsection{Assignment games with two row and two column players}

Let us summarize first what we know about the four-player assignment games with no veto player.
They must contain at least two row and at least two column players, none of them is a null player.
Thus, the underlying surplus matrix must be one of the following three types of $2\times 2$ matrices:
$$
\begin{array}{c|cc|}
\textbf{A} & 3 & 4 \\
\hline
1 & + & 0  \\
2 & 0 & +  \\
\hline
\end{array}
\qquad\qquad
\begin{array}{c|cc|}
\textbf{B} & 3 & 4 \\
\hline
1 & + & +  \\
2 & + & 0  \\
\hline
\end{array}
\qquad\qquad
\begin{array}{c|cc|}
\textbf{C} & 3 & 4 \\
\hline
1 & + & +  \\
2 & + & +  \\
\hline
\end{array}
$$

By Propossition~\ref{prop:assignment-convex}, Type A assignment games are convex games, so by Theorem~\ref{theo:Sprumont1990} they admit a PMAS and all their core allocations are PMAS-extendable. 

In contrast, by Theorem~\ref{theo:Sprumont1990}, Type C assignment games admit no PMAS.

Next we fill the gap and show that Type B assignment games, although not convex games, also admit a PMAS and all their core allocations are PMAS-extendable.

\begin{proposition}
\label{prop:assign-2x2}
Let assignment game $(I\cup J, w)$ with row players $I=\{1,2\}$ and column players $I=\{3,4\}$ be induced by pairwise surplus matrix  
$\begin{array}{c|cc|}
  & 3 & 4 \\
\hline
1 & a & b \\
2 & c & 0 \\
\hline
\end{array}$, where $a,b,c>0$. Then

\textbf{Case 1:} if $a+0\geq b+c$ then all core allocations are PMAS-extendable; 

\textbf{Case 2:} if $b+c>a+0$ then the game admits no PMAS.
\end{proposition}
\begin{proof}
In \textbf{Case 1}, $w(I\cup J)=a+0=w(\{1,3\})$.
Since players 2 and 4 form an inessential pair, they get 0 payoff in any core allocation. 
Thus, the core is the set of allocations $\{(x_1,0,x_3,0) : x_1+x_3=a, x_1\geq b, x_3\geq c\}$.
It follows from Proposition~\ref{prop:inessential-pmas} that the following simplified scheme (containing only the multi-player essential coalitions and the grand coalition) determines a PMAS which extends any arbitrarily chosen core allocation $(x_1,0,x_3,0)$:
$$
\setlength{\extrarowheight}{2pt}
\begin{array}{|cc|cc||c|c|cc|cc|}
\multicolumn{4}{|c||}{S} & w(S) & 0\leq & p^S_1 & p^S_2 & p^S_3 & p^S_4 \\
\hline
\hline
1 & . & 3 & . & 0<a & = & x_1 & . & x_3 & . \\
1 & . & . & 4 & 0<b & = & b & . & . & 0 \\
. & 2 & 3 & . & 0<c & = & . & 0 & c & . \\
\hline
1 & 2 & 3 & 4 & a+0 & = & x_1 & 0 & x_3 & 0 \\
\hline
\end{array}
$$
Notice that for any core allocation there is a unique PMAS which extends it, for $(x_1,0,x_3,0)$ the above one.
\medskip

In \textbf{Case 2}, $w(I\cup J)=b+c=w(\{1,4\})+w(\{2,3\})$.
Thus, the core is the set of allocations $\{(x_1,x_2,x_3,x_4) : x_1+x_4=b, x_2+x_3=c, x_1+x_3\geq a, x_1, x_2, x_3, x_4\geq 0\}$.

Suppose there is a PMAS for $w$. Then the following relations must hold:
\begin{itemize}
\item by Proposition~\ref{prop:inessential-pmas}, $p^{1234}_1=p^{14}_1$, $p^{1234}_4=p^{14}_4$, and $p^{1234}_2=p^{23}_2$, $p^{1234}_3=p^{23}_3$;
\item by individual payoff monotonicity, $p^{134}_1=p^{14}_1$, $p^{134}_4=p^{14}_4$, and $p^{123}_2=p^{23}_2$, $p^{123}_3=p^{23}_3$;
\item then from $w(\{1,2,3\})=a\vee c$ and $w(\{1,3,4\})=a\vee b$, by suballocation efficiency, $p^{134}_3=(a-b)^+$ and $p^{123}_1=(a-c)^+$,
\end{itemize}
where $a\vee b$ shorthands $\max\{a,b\}$ and $(a-b)^+$ shorthands $\max\{a-b,0\}$ for any two real numbers $a$ and $b$. 

Then the corresponding suballocations in the PMAS must look like this:
$$
\setlength{\extrarowheight}{2pt}
\begin{array}{|@{\;}c@{\;}c@{\;}|@{\;}c@{\;}c@{\;}||c|c||cc|cc||c}
\multicolumn{4}{|c||}{S} & w(S) & 0\leq & p^S_1 & p^S_2 & p^S_3 & p^S_4 \\
\hline
\hline
1 & . & 3 & . & 0<a & = & p^{13}_1 & . & p^{13}_3 & . \\
1 & . & . & 4 & 0<b & = & p^{14}_1 & . & . & p^{14}_4 \\
. & 2 & 3 & . & 0<c & = & . & p^{23}_2 & p^{23}_3 & . \\
\hline
1 & 2 & 3 & . & a\vee c & = & p^{123}_1 & p^{23}_2 & p^{23}_3 & . & p^{123}_1=(a-c)^+ \\
1 & . & 3 & 4 & a\vee b & = & p^{14}_1 & . & p^{134}_3 & p^{14}_4 & p^{134}_3=(a-b)^+ \\
\hline
1 & 2 & 3 & 4 & b+c & = & p^{14}_1 & p^{23}_2 & p^{23}_3 & p^{14}_4 \\
\hline
\end{array}
$$
From suballocation efficiency and individual payoff monotonicity we get $0<a=p^{13}_1+p^{13}_3\leq p^{123}_1+p^{134}_3 = (a-c)^+ + (a-b)^+ $.
Since $a,b,c>0$, so $(a-c)^+<a$ and $(a-b)^+<a$, both $(a-c)^+=a-c$ and $(a-b)^+=a-b$ must hold.
Then we get the contradiction $0<a \leq (a-c)^+ + (a-b)^+ = (a-c)+(a-b) = a+(a-c-b) < a$. 
Therefore, in case of $b+c>a+0$, the assignment game $w$ is not PMAS-admissible.
\end{proof}

\subsection{Assignment games with a veto mixed pair}

We say that a nonnegative matrix $A=\big[a_{ij}\geq 0\big]_{i\in I, j\in J}$ with $\vert I\vert,\vert J\vert\geq 2$ has a \emph{$\Gamma$-shaped support}, if there exist a row player $i_1\in I$ and a column player $j_1\in J$ such that $a_{kl}>0$ if $k=i_1$ and/or $l=j_1$, but  $a_{kl}=0$ whenever $k\neq i_1$ and $l\neq j_1$.
We call the entry $a_{i_1 j_1}$ the \emph{corner} of the $\Gamma$-shaped support, and say that it is \emph{dominant}, if $a_{i_1 j_1}\geq a_{i_1 l}+a_{k j_1}$ whenever $k\neq i_1$ and $l\neq j_1$.

Notice that in the assignment game induced by a nonnegative matrix with $\Gamma$-shaped support none of the players is a veto player or a null player.
The corner pair coalition $\{i_1,j_1\}$, however, has veto power, the worth of a coalition is zero unless it contains both $i_1$ and $j_1$.
For example, the support of matrix 
$$
\begin{array}{c|cccc|}
  & j_1 & j_2 & \ldots & j_{\vert J\vert} \\
\hline
i_1 & a_{11} & a_{12} & \ldots & a_{1\vert J\vert} \\
i_2 & a_{21} & 0      & \ldots & 0   \\
\vdots & \vdots  & \vdots &  & \vdots \\
i_{\vert I\vert} & a_{\vert I\vert1} & 0 & \ldots & 0  \\
\hline
\end{array}
$$
is $\Gamma$-shaped if $a_{ij}>0$ for all $1\leq i\leq \vert I\vert$ and $1\leq j\leq \vert J\vert$. 
Its corner is dominant if $a_{11}\geq a_{1l}+a_{k1}$ for all $2\leq k\leq \vert I\vert$ and $2\leq l\leq \vert J\vert$. 
In the corresponding assignment game the corner pair coalition $\{i_1,j_1\}$ has veto power, although neither $i_1$ nor $j_1$ is a veto player.
In the sequel we assume (without loss of generality) that the $\Gamma$-shaped support of the underlying matrix is composed of the first row and the first column, and starting from the corner the other entries are arranged in a weakly decreasing order in both directions, that is $a_{1l}\geq a_{1l'}$ for all $2\leq l<l'\leq \vert J\vert$, and $a_{k1}\geq a_{k'1}$ for all $2\leq k<k'\leq \vert I\vert$.

As expected, the characterization of the $2\times 2$ type B assignment games (Proposition~\ref{prop:assign-2x2}) extends to larger games induced by matrices with $\Gamma$-shaped support.
\begin{proposition}
\label{prop:assign-Gamma-shaped}
An assignment game induced by a matrix with $\Gamma$-shaped support admits a PMAS if and only if its corner is dominant.
\\
Moreover, if the game is PMAS-admissible then every core allocation is PMAS-extendable.
\end{proposition}
\begin{proof}
We prove sufficiency in an indirect way by using the characterization by \cite{Sprumont1990} (see Theorem~\ref{theo:Sprumont1990}), although a direct constructive proof akin to that of Case~1 in Proposition~\ref{prop:assign-2x2} could also be given.

Let assignment game $(I\cup J,w)$ be induced by a matrix with $\Gamma$-shaped support whose corner is dominant. 
Then this matrix can be written as the sum of two nonnegative matrices, one with a single positive row, the other with a single positive column, as follows:
$$
\begin{array}{|cccc|}
\hline
a_{11} & a_{12} & \ldots & a_{1\vert J\vert} \\
a_{21} & 0      & \ldots & 0   \\
\vdots  & \vdots &  & \vdots \\
a_{\vert I\vert 1} & 0 & \ldots & 0  \\
\hline
\end{array}
=
\begin{array}{|cccc|}
\hline
 a_{11}-a_{21} & a_{12} & \ldots & a_{1\vert J\vert} \\
 0 & 0      & \ldots & 0   \\
\vdots  & \vdots &  & \vdots \\
 0 & 0 & \ldots & 0  \\
\hline
\end{array}
+
\begin{array}{|cccc|}
\hline
 a_{21} & 0 & \ldots & 0 \\
 a_{21} & 0      & \ldots & 0   \\
 \vdots  & \vdots &  & \vdots \\
 a_{\vert I\vert 1} & 0 & \ldots & 0  \\
\hline
\end{array}.
$$
Both matrices on the right induce an assignment game with a single veto player: the first row player $i_1$ in the first game (denoted by $w^1$), and the first column player $j_1$ in the second game (denoted by $w^2$).
It is easily checked that $w=w^1+w^2$. For example, since the corner is dominant, $w(\{i_1,i_2,j_1,j_2\})=a_{11}=(a_{11}-a_{21})+a_{21}=w^1(\{i_1,i_2,j_1,j_2\})+w^2(\{i_1,i_2,j_1,j_2\})$. 
Since assignment games are 0-normalized monotonic games, by Theorem~\ref{theo:Sprumont1990} we get that our game $w$, as a positive linear combination of monotonic games with at least one veto player, is PMAS-admissible.

In this dominant corner case, notice that in core payoffs only the players $i_1$ and $j_1$ will get positive payoff, since all other optimally matched pairs are of zero value. 
Take any allocation $y=(u_1,0,\ldots,0;\; v_1,0,\ldots,0)$ in the core of $w$. The positive payoffs satisfy $u_1+v_1=a_{11} $, $u_1\geq a_{12}$, and $v_1\geq a_{21}$.
As in the proof of Proposition~\ref{prop:assign-veto}, it is easily checked that the payoff vector $y^1=(u_1,0,\ldots,0;\; v_1-a_{21},0,\ldots,0)$ is a core allocation in $w^1$, and the payoff vector $y^2=(0,0,\ldots,0;\; a_{21},0,\ldots,0)$ is a core allocation in $w^2$. Thus, in this special case the core mapping is additive, that is, $\mathbf{C}(w)=\mathbf{C}(w^1)+\mathbf{C}(w^2)$ (not just superadditive as it is in general, see \citep{MiquelNunez2011} for other kinds of sufficient conditions for additivity of the core on assignmet games). Since, by Proposition~\ref{prop:assign-veto}, core allocation $y^1$ is PMAS-extendable in $w^1$ and core allocation $y^2$ is PMAS-extendable in $w^2$, we get that the arbitrarily chosen core allocation $y=y^1+y^2$ is PMAS-extendable in $w=w^1+w^2$.  

Finally, if the corner of the $\Gamma$-shaped support is not dominant than it is not dominant in the upper-left $2\times 2$ submatrix, so by Case 2 in Proposition~\ref{prop:assign-2x2}, the corresponding subgame has no PMAS, hence the game itself cannot be PMAS-admissible. 
\end{proof}

\section{The main structural result}\label{mainresult}

We prove that any PMAS-admissible assignment game is the composition of the three special types of PMAS-admissible assignment games discussed in the previous section, and also give a characterization in terms of the properties of the underlying pairwise surplus matrix.  
\begin{theorem}
\label{theo:assign-PMAS-char}
The following statements are equivalent. 
\\ An assignment game ...
\begin{enumerate}
\item admits a PMAS;
\item is induced by a block-diagonal nonnegative matrix with blocks having support of one of the following three types: (i) single row vector, (ii) single column vector, (iii) $\Gamma$-shaped with dominant corner;
\item is decomposable in assignment games with veto player or with dominant veto mixed pair. 
\end{enumerate}
Moreover, if an assignment game admits a PMAS then every core allocation can be extended to a PMAS.
\end{theorem}
\begin{proof}
We show the chain of implications $1.\Rightarrow 2.\Rightarrow 3.\Rightarrow 1.$ for assignment game $(I\cup J,w)$ induced by nonnegative matrix $A$ with $\vert I\vert$ rows and $\vert J\vert$ columns.
If $\vert I\vert=1$ or $\vert J\vert=1$ or $\vert I\vert=\vert J\vert=2$, all statements have been proved in Propositions~\ref{prop:assign-veto} and \ref{prop:assign-2x2}.
Thus, we assume that $A$ contains a $2\times 3$ and/or a $3\times 2$ submatrix. 
Furthermore, since a null player receives zero payoff in any core allocation and in any suballocation of a PMAS, we assume (without loss of generality) that none of the players is a null player, or equivalently, that in each row and column of $A$ there is at least one positive entry.
\medskip

$1.\Rightarrow 2.$\quad
Suppose $w$ is PMAS-admissible. 
We first show the following 
\begin{itemize}
\item \emph{Claim: if $w$ is PMAS-admissible and a row (column) of the underlying matrix $A$ contains 2 positive entries then only at most one of the two corresponding columns (rows) can contain other positive entries.}
\end{itemize}
Suppose not, and in the column of $a_{11}>0$ as well as in the column of $a_{12}>0$ there is another positive entry, say $a_{j1}>0$ and  $a_{k2}>0$, respectively. 
If $j=k$ then $A$ contains a $2\times 2$ submatrix with 4 positive entries, so by Theorem~\ref{theo:Sprumont1990}, the corresponding $(2+2)$-player subgame admits no PMAS, a contradiction.
If $j\neq k$ then in the $2\times 2$ submatrix with rows 1 and $j$ and columns 1 and 2, by Proposition~\ref{prop:assign-2x2}, $a_{11}\geq a_{12}+a_{j1}$ and $a_{j2}=0$ must hold, otherwise the corresponding $(2+2)$-player subgame would not admit a PMAS, a contradiction.
Similarly, in the $2\times 2$ submatrix with rows 1 and $k$ and columns 1 and 2, $a_{12}\geq a_{11}+a_{k2}$ and $a_{k1}=0$ must hold. These two inequalities, however, would imply $a_{11}> a_{12}>a_{11}$, again a contradiction, proving the Claim.

Now, arrange the columns of matrix $A$ so that the entries in the first row come in a nonincreasing order from left to right, ties are broken arbitrarily.

If there are at least two positive entries in the first row, the Claim implies that at most one of the corresponding columns can contain positive entries, and they are the only positive entries in their row. If one of these columns has at least two positive entries, the submatrix spanned by all these positive entries has a $\Gamma$-shaped support whose corner must be dominant by Proposition~\ref{prop:assign-Gamma-shaped}, thus it must be $a_{11}$.
In this subcase, we rearrange the rows so that the entries in the first column come in a nonincreasing order, and thus get a submatrix of $\Gamma$-shaped support with dominant corner in the upper-left corner of $A$.  
In the other subcase, that is when there is no other positive entry in any of the columns of the positive entries in the first row, we get a single row positive submatrix in the upper-left corner of $A$.

If there is only one positive entry in the first row, we do the above routine by interchanging the roles of rows and columns, and get either a submatrix of $\Gamma$-shaped support with dominant corner or a single column positive submatrix in the upper-left corner of $A$.

When we delete from $A$ all rows and columns of the positive entries in the previously constructed submatrix, we do not delete any other positive entries of $A$. Thus, when we repeat the above routine for the remaining submatrix $A'$ of $A$, the support of the newly obtained submatrix in the upper-left corner of $A'$ must be either $\Gamma$-shaped with dominant corner or a single row or a single column.
At the end of this iterative rearrangement and deletion process matrix $A$ gets the block-diagonal structure stated in item $2.$ that is illustrated below
$$
\setlength{\extrarowheight}{2pt}
\begin{array}{|@{\;}c@{\;}c@{\;}c@{\;}|@{\;}c@{\;}c@{\;}c@{\;}|@{\;}c@{\;}|}
\hline
\fbox{+} & + & + &  &  &  &  \\
 + & 0 & 0 &  &  &  &  \\
\hline
  &  &  & \fbox{+} & + & + &  \\[2pt]
\hline
  &  &  &  &  &  & \fbox{+} \\
  &  &  &  &  &  & + \\
\hline
\end{array}
$$
where the dominant corners and entries of the block supports are boxed and the full 0 blocks are left empty.   
\medskip

$2.\Rightarrow 3.$\quad The implication is a straightforward consequence of the definitions of an assignment game, a veto player, and the composition of games.
\medskip

$3.\Rightarrow 1.$\quad The implication as well as the PMAS-extendability of every core allocation in PMAS-admissible assignment games follow from Propositions~\ref{prop:inessential-pmas}, \ref{prop:assign-Gamma-shaped}, and \ref{prop:assign-veto}.
\end{proof}

\section{The nucleolus and the tau-value}\label{nucleolus=tau}

We investigate the relation of the nucleolus and the tau-value in PMAS-admissible assignment games.
We find that, although they are ``typically'' different in assignment games, but they coincide on the PMAS-admissible subclass.

The tau-value was introduced by \cite{Tijs1981} as an efficient compromise value between an upper vector and a lower vector which, in balanced games, specify (not necessarily sharp) upper, respectively lower bounds for core payoffs. 
The \emph{upper vector} in game $(N,v)$ consists of the marginal contribution of the players to the grand coalition, $M_i(v)=v(N)-v(N\setminus\{i\})$ for all $i\in N$.
The \emph{lower vector} is defined by $m_i(v)=\max_{S: S\ni i} \left( v(S)-\sum_{j\in S\setminus\{i\}} M_j(v) \right)$ for all $i\in N$.
In balanced game $(N,v)$, for any $x\in\mathbf{C}(v)$, the relations $m_i(v)\leq x_i\leq M_i(v)$ hold for all $i\in N$.
Thus, there is a unique number $0\leq \kappa\leq 1$ for which $\sum_{i\in N} \tau_i(v)=v(N)$, where $\tau_i(v):=\kappa M_i(v)+(1-\kappa) m_i(v)$ is called the \emph{tau-value} of player $i\in N$ in balanced game $(N,v)$.
In line with the standard terminology, the tau-value of the game is the vector of the individual tau-values of the players, $\tau(v)=\left( \tau_i(v): i\in N \right)$.

In an arbitrary balanced game the tau-value need not be a core allocation.
\cite{Driessen1988} gave necessary and sufficient conditions for the tau-value to lie in the core. 
He also demonstrated that, unlike the core and the nucleolus, the tau-value is not independent of the inessential coalitions by presenting two different games with identical cores, but different tau-values, one of them even lying outside the core.
For assignment games, \cite{NunezRafels2002} proved that the tau-value is the midpoint of the line segment joining the row-optimal and the column-optimal core vertices, hence it is a core allocation, moreover, it is a locus of the core.

The nucleolus was introduced by \cite{Schmeidler1969}. For balanced games, it is the unique core allocation which lexicographically maximizes the nondecreasingly ordered vector of the satisfactions of coalitions over the core, where we define the \emph{satisfaction} of coalition $S$ at allocation $x$ as $f(S,x)=x(S)-v(S)$.  
\cite{Huberman1980} showed that for balanced games, the essential coalitions suffice to determine the nucleolus, the inessential coalitions can be ignored, just like for the core.
Consequently, for balanced games, the general criterion of \cite{Kohlberg1971} can also be simplified to contain only the essential coalitions as follows.
\begin{lemma}
\label{lemma:Kohlberg}
In balanced game $(N,v)$, the core allocation $x\in\mathbf{C}(v)$ is the nucleolus if and only if for any real number $t$ the family of essential coalitions 
$$\mathcal{F}_t(x)=\left\{S\in\mathcal{E}(v)\setminus\{N\}: x(S)-v(S)\leq t \right\}$$
is either empty or balanced on $N$ (i.e., there are positive numbers $\lambda_S>0$ for all $S\in\mathcal{F}_t(x)$ such that $\sum_{S\ni i,S\in\mathcal{F}_t(x)}\lambda_S =1$ for all $i\in N$).
\end{lemma}
We use this criterion in two balanced games to check whether the tau-value is the nucleolus.
The first one is the veto-controlled monotonic, hence PMAS-admissible game of Example~\ref{exam:4player-veto}.
The second one is a $(2+2)$-player assignment game with no PMAS as discussed in Case 2 of Proposition~\ref{prop:assign-2x2}.

\begin{example}[PMAS-admissible game where the tau-value and the nucleolus are different]
\label{exam:4player-veto-tau}

Consider the game from Example~\ref{exam:4player-veto} with player set $N=\{1,2,3,4\}$ and coalitional function $v(12)=2$, $v(13)=v(123)=3$, $v(14)=v(124)=4$, $v(134)=5$, $v(N)=8$, and $v(T)=0$ for any other coalition $T\subseteq N$.
It is easily computed that the upper vector is $M(v)=[8,3,4,5]$, the lower vector is $m(v)=[0,0,0,0]$, so the efficiency scalar is $\kappa=8/20=2/5$.
Thus, the tau-value is $\tau(v)=[16/5,6/5,8/5,10/5]$.
The payoffs and the satisfactions for the essential coalitions are listed, respectively, in columns $\tau(S)$ and $\tau(S)-v(S)$ of the table below.
Since all satisfactions are positive (except for the grand coalition, of course, as it must always be 0), the tau-value is in the core (in fact, in its relative interior). 
For the smallest satisfaction $t=6/5$ we have $\mathcal{F}_{6/5}(\tau)=\big\{ \{2\}, \{1,4\} \big\}$.
It is not a balanced family (it does not even cover player 3), thus, the tau-value cannot be the nucleolus.  
$$
\setlength{\extrarowheight}{1pt}
\begin{array}{|cccc||c||cc||cc|}
\multicolumn{4}{|c||}{S} & v(S) & \tau(S) & \tau(S)-v(S) & \eta(S) & \eta(S)-v(S) \\
\hline
\hline
1 & . & . & . & 0 & 16/5 & 16/5 & 21/6 & 21/6 \\
. & 2 & . & . & 0 &  6/5 &  6/5 &  8/6 &  8/6 \\
. & . & 3 & . & 0 &  8/5 &  8/5 &  8/6 &  8/6 \\
. & . & . & 4 & 0 & 10/5 & 10/5 & 11/6 & 11/6 \\
\hline
1 & 2 & . & . & 2 & 22/5 & 12/5 & 29/6 & 17/6 \\
1 & . & 3 & . & 3 & 24/5 &  9/5 & 29/6 & 11/6 \\
1 & . & . & 4 & 4 & 26/5 &  6/5 & 32/6 &  8/6 \\
1 & . & 3 & 4 & 5 & 34/5 &  9/5 & 40/6 & 10/6 \\
\hline
1 & 2 & 3 & 4 & 8 & 8 & 0 & 8 & 0 \\
\hline
\end{array}
$$
Indeed, if we compute the nucleolus (for example, with the specialized method of \cite{Bahel2019}), we get $\eta(v)=[21/6,8/6,8/6,11/6]$.
The payoffs and the satisfactions for the essential coalitions are given, respectively, in columns $\eta(S)$ and $\eta(S)-v(S)$ of the table above.
For the smallest satisfaction $t=8/6$, the family $\mathcal{F}_{8/6}(\eta)=\big\{ \{2\}, \{3\}, \{1,4\} \big\}$ is a partition of $N$, hence it is balanced with weights 1.
For the second smallest satisfaction $t=10/6$, the family $\mathcal{F}_{10/6}(\eta)=\big\{ \{2\}, \{3\}, \{1,4\}, \{1,3,4\} \big\}$ is also balanced, as it is the union of partitions.
The same holds for the family $\mathcal{F}_{11/6}(\eta)=\big\{ \{2\}, \{3\}, \{1,4\}, \{1,3,4\}, \{4\}, \{1,3\} \big\}$ corresponding to the third smallest satisfaction $t=11/6$. And so on, for any higher values of $t$.
Therefore, by the Kohlberg-criterion, the core allocation $\eta(v)$ is indeed the nucleolus in this game. 
\examend
\end{example}

\begin{example}[PMAS-inadmissible assignment game where the tau-value and the nucleolus are different]
\label{exam:4player-assign-tau}

Consider the assignment game with player set $\{1,2\}\cup\{3,4\}$ induced by the following matrix with $\Gamma$-shaped support whose corner is not dominant (although the largest entry):
$\begin{array}{c|cc|}
 & 3 & 4 \\
\hline
1 & 6 & 3  \\
2 & 5 & 0  \\
\hline
\end{array}$.

The coalitional function is $w(13)=w(134)=w(123)=6$, $w(14)=w(124)=3$, $w(23)=w(234)=5$, $w(N)=8$, and $w(T)=0$ for any other coalition $T\subseteq N$.
It is easily computed that the upper vector is $M(w)=[3,2,5,2]$, the lower vector is $m(w)=[1,0,3,0]$, so the efficiency scalar is $\kappa=(8-4)/(12-4)=1/2$.
Thus, the tau-value is $\tau(w)=[2,1,4,1]$.
The payoffs and the satisfactions for the essential coalitions are listed, respectively, in columns $\tau(S)$ and $\tau(S)-w(S)$ of the table below.
For the smallest satisfaction $t=0$ we have $\mathcal{F}_{0}(\tau)=\big\{ \{1,4\}, \{2,3\}, \{1,3\} \big\}$.
It is not a balanced family (although it contains a partition), thus, the tau-value cannot be the nucleolus.  
$$
\setlength{\extrarowheight}{1pt}
\begin{array}{|cccc||c||cc||cc|}
\multicolumn{4}{|c||}{S} & w(S) & \tau(S) & \tau(S)-w(S) & \eta(S) & \eta(S)-w(S) \\
\hline
\hline
1 & . & . & . & 0 & 2 & 2 &  7/3 &  7/3 \\
. & 2 & . & . & 0 & 1 & 1 &  2/3 &  2/3 \\
. & . & 3 & . & 0 & 4 & 4 & 13/3 & 13/3 \\
. & . & . & 4 & 0 & 1 & 1 &  2/3 &  2/3 \\
\hline
1 & . & 3 & . & 6 & 6 & 0 & 20/3 &  2/3 \\
1 & . & . & 4 & 3 & 3 & 0 &    3 &    0 \\
. & 2 & 3 & . & 5 & 5 & 0 &    5 &    0 \\
\hline
1 & 2 & 3 & 4 & 8 & 8 & 0 & 8 & 0 \\
\hline
\end{array}
$$
Indeed, if we compute the nucleolus (for example, with the specialized algorithm of \cite{SolymosiRaghavan1994}), we get $\eta(w)=[7/3,2/3,13/3,2/3]$.
The payoffs and the satisfactions for the essential coalitions are given, respectively, in columns $\eta(S)$ and $\eta(S)-w(S)$ of the table above.
For the smallest satisfaction $t=0$, the family $\mathcal{F}_{0}(\eta)=\big\{ \{1,4\}, \{2,3\}\big\}$ is a partition of $N$, hence it is balanced with weights 1.
For the second smallest satisfaction $t=2/3$, the family $\mathcal{F}_{2/3}(\eta)=\big\{ \{1,4\}, \{2,3\}, \{2\}, \{4\}, \{1,3\} \big\}$ is also balanced, as it is the union of partitions.
The same holds for the families correspond to any higher value of $t$.
Therefore, by the Kohlberg-criterion, the core allocation $\eta(w)$ is indeed the nucleolus in this game. 
\examend
\end{example}

We show that if the assignment game is PMAS-admissible then the tau-value and the nucleolus coincide. 

\begin{theorem}
\label{theo:tau=nucleolus}
In PMAS-admissible assignment games the tau-value and the nucleolus coincide.
\end{theorem}
\begin{proof}
First we determine the tau-value of a PMAS-admissible assignment game, then we apply the specialized version of the  Kohlberg-criterion to show that it is the nucleolus.
Let assignment game $(I\cup J,w)$ be PMAS-admissible.
By Theorem~\ref{theo:assign-PMAS-char}, it is induced by a block-diagonal nonnegative matrix $A$ with three possible types of blocks.

As proved by \cite{NunezRafels2002}, the tau-value in an assignment game is the simple average of the row-optimal and the column-optimal core extreme points.
Since the core of a composite totally balanced game is always the Cartesian product of the cores of the component balanced games, it is easily seen that the row-optimal and the column-optimal core extreme points of our assignment game $w$ are, respectively, the Cartesian products of the row-optimal and the column-optimal core extreme points of the component assignment games.

Let us see these special core vertices in the three special types of assignment games appearing in Theorem~\ref{theo:assign-PMAS-char}. For simplicity, we use as few generic indices as necessary.
\begin{enumerate}
\item Assignment game $w$ with a single row player induced by a row vector.
Recall from the proof of Proposition~\ref{prop:assign-veto} that for the core only the two highest surplusses matter, any other column player receives fixed zero payoff over the core.
Thus, we assume that $I=\{1\}$, $J=\{2,3,4\}$, and $w$ is induced by row vector $[a, b, c]$ where $a \geq b \geq c >0$.
The row-optimal vertex of the core is $(\overline{x}_1=a,\underline{x}_2=0,\underline{x}_3=0,\underline{x}_4=0)$.
The column-optimal vertex is $(\underline{x}_1=b,\overline{x}_2=a-b,\overline{x}_3=0,\overline{x}_4=0)$.
Thus, the tau-value is $\tau(w)=(\tau_1=\frac{a+b}{2},\tau_2=\frac{a-b}{2},\tau_3=0,\tau_4=0)$.
The satisfactions at the tau-value are the payoffs for the single-player coalitions and $f(\{1,2\},\tau)=0$, $f(\{1,3\},\tau)=\frac{a-b}{2}$, $f(\{1,4\},\tau)=\frac{a-b}{2}+b-c$ for the other essential coalitions.
The lowest satisfaction is $t=0$, the family $\mathcal{F}_{0}(\tau)=\big\{ \{1,2\}, \{3\}, \{4\} \big\}$ is balanced.
The second lowest satisfaction is $t=\frac{a-b}{2}$, the family $\mathcal{F}_{\frac{a-b}{2}}(\tau)=\big\{ \{1,2\}, \{3\}, \{4\}, \{2\}, \{1,3\} \big\}$ is also balanced, as it is the union of partitions.
Since already this family is rich enough to uniquely determine the core allocation which could provide these two lowest satisfaction levels, by Lemma~\ref{lemma:Kohlberg}, the tau-value is the nucleolus.
\item The coincidence of the tau-value and the nucleolus for an assignment game with a single column player induced by a column vector is seen in the same way.
\item Assignment game $w$ with dominant veto mixed pair. 
Recall from the proof of Proposition~\ref{prop:assign-Gamma-shaped} that for the core only the two highest surplusses in the row and the column of the corner matter, any other row or column player receives fixed zero payoff over the core.
Thus, we assume that $I=\{1,2,3\}$, $J=\{4,5,6\}$, and $w$ is induced by the matrix
$
\begin{array}{c|ccc|}
  & 4 & 5 & 6 \\
\hline
1 & a & b & d \\
2 & c & 0 & 0 \\
3 & f & 0 & 0 \\
\hline
\end{array}
$
where $a \geq b+c$, $b \geq d >0$, and $c \geq f >0$.
The row-optimal vertex of the core is $(\overline{x}_1=a-c,\overline{x}_2=0,\overline{x}_3=0;\, \underline{x}_4=c,\underline{x}_5=0,\underline{x}_6=0)$.
The column-optimal vertex is $(\underline{x}_1=b,\underline{x}_2=0,\underline{x}_3=0;\, \overline{x}_4=a-b,\overline{x}_5=0,\overline{x}_6=0)$.
Thus, the tau-value is $\tau(w)=(\tau_1=\frac{a+b-c}{2},\tau_2=0,\tau_3=0;\, \tau_4=\frac{a-b+c}{2},\tau_5=0,\tau_6=0)$.
The satisfactions at the tau-value are the payoffs for the single-player coalitions and $f(\{1,4\},\tau)=0$, $f(\{1,5\},\tau)=f(\{2,4\},\tau)=\frac{a-b-c}{2}$, $f(\{1,6\},\tau)=\frac{a-b-c}{2}+b-d$, $f(\{3,4\},\tau)=\frac{a-b-c}{2}+c-f$ for the other essential coalitions.
The lowest satisfaction is $t=0$, the family $\mathcal{F}_{0}(\tau)=\big\{ \{1,4\}, \{2\}, \{3\}, \{5\}, \{6\} \big\}$ is balanced.
The second lowest satisfaction is $t=\frac{a-b-c}{2}$, the family $\mathcal{F}_{\frac{a-b-c}{2}}(w)=\mathcal{F}_{0}(\tau)\cup \big\{ \{1,5\}, \{2,4\} \big\}$ is also balanced, as it is the union of partitions.
Since already this family is rich enough to uniquely determine the core allocation which could provide these two lowest satisfaction levels, by Lemma~\ref{lemma:Kohlberg}, the tau-value is the nucleolus.
\end{enumerate}
Notice that in each of these special assignment games the lowest satisfaction at the tau-value is $t=0$, that is achieved by a partition composed of the dominant mixed pair and the singleton coalitions of the other players.
Then it is easily seen that in a composition of these special assignment games, the family corresponding to any positive satisfaction level is the union of the respective families in the component games, each of which is balanced on its own player set. Clearly, the union of these balanced families is balanced on the player set of the composite game.
Therefore, by Lemma~\ref{lemma:Kohlberg}, also in the composite game the tau-value coincides with the nucleolus.
\end{proof}

The decomposibility of both the nucleolus and the tau-value on the class of assignment games are key elements of the above proof. 
It could be easily seen that the decomposibility of the nucleolus extends to the class of totally balanced games.
The following example shows that this is not the case for the tau-value, as it fails to be decomposable on  convex games, an important subclass of totally balanced games. 
\begin{example}[The tau-value is not decomposable on the class of convex games]
\label{exam-tau-nondecomp}
Consider the following two games 
$$
\begin{array}{c||ccc|ccc|c}
S      & 1 & 2 & 3 & 12 & 13 & 23 & 123 \\
\hline 
v_1(S) & 0 & 0 & 0 &  1 &  2 &  3 &   5 \\
\end{array}
\qquad\qquad
\begin{array}{c||cc|c}
T      & 4 & 5 & 45 \\
\hline 
v_2(T) & 0 & 0 &  3 \\
\end{array}.
$$
Clearly, both  $(N_1,v_1)$ and $(N_2,v_2)$ are convex games, hence their composite 5-player game, $(N,v)$ is also a convex game.
The following three statements, which help to justify our computation, are easily seen:
For convex games, 
\begin{itemize}
\item the upper vector of the composite game is the Cartesian product of the upper vectors of the component games (in fact, this holds for all games);
\item the lower vector of the composite game is the Cartesian product of the lower vectors of the component games (in fact, this holds for a larger class games);
\item the lower vector is the vector of the individual values (in fact, this holds for a larger class games). 
\end{itemize}
We get that the upper vector, the lower vector, and the tau-value for the above games $(N_1,v_1)$ and $(N_2,v_2)$ are
$$
\begin{array}{c||ccc|}
i\in N_1 & 1 & 2 & 3 \\
\hline 
M_i(v_1) & 2 & 3 & 4 \\
\hline 
m_i(v_1) & 0 & 0 & 0 \\
\hline
\hline 
\tau_i(v_1) & 10/9 & 15/9 & 20/9 \\
\end{array}
\qquad\qquad
\begin{array}{c||ccc|}
j\in N_2 & 4 & 5 \\
\hline 
M_j(v_2) & 3 & 3 \\
\hline 
m_j(v_2) & 0 & 0 \\
\hline
\hline 
\tau_j(v_2) & 3/2 & 3/2 \\
\end{array}.
$$ 
For the composite game $(N,v)$ we get 
$$
\begin{array}{c||ccccc|}
k\in N & 1 & 2 & 3 & 4 & 5 \\
\hline 
M_k(v) & 2 & 3 & 4 & 3 & 3 \\
\hline 
m_k(v) & 0 & 0 & 0 & 0 & 0 \\
\hline
\hline 
\tau_k(v) & 16/15 & 24/15 & 32/15 & 24/15 & 24/15 \\
\end{array},
$$ 
and observe that $\tau(v)\neq\tau(v_1)\times\tau(v_2)$.

In contrast, for the nucleolus, we have 
$\eta(v) = [1,\, 3/2,\, 5/2,\, 3/2,\, 3/2] = [1,\, 3/2,\, 5/2] \times [3/2,\, 3/2] = \eta(v_1)\times\eta(v_2)$.  
\examend
\end{example}

\section{Conclusion}\label{conlusion}

We have characterized the assignment games which admit a population monotonic allocation scheme in terms of efficiently verifiable properties of the pairwise surplus matrix that induces the assignment game.
We have proved that an assignment game is PMAS-admissible if and only if the underlying nonnegative matrix, up to row and column permutations, is a block-diagonal matrix where each block is either a single row vector, or a single column vector, or has a $\Gamma$-shaped support whose corner element itself dominates the value of any other assignment in the block.
Translated in game theoretic terms it means that an assignment game is PMAS-admissible if and only if it contains a veto (row or column) player or a dominant veto mixed pair or is composed of from these two types of special assignment games.

We have also shown that PMAS-admissibility of the assignment game implies that all core allocations can be extended to a PMAS, similarly to convex games, but unlike in arbitrary veto-controlled monotonic games where some core allocations might not be PMAS-extendable.
Finally, we have proved that in a PMAS-admissible assignment game, the nucleolus and the tau-value coincide.

To facilitate our discussions, we have made some (to best of our knowledge new) useful general observations. 
We have shown, for example, that population monotonic allocation schemes are independent of inessential coalitions, and are decomposable on the class of totally balanced games, exactly like the core and the nucleolus, but unlike the tau-value that, as we have illustrated, is not decomposable on the class of convex games, hence cannot be independent of inessential coalitions either.
In light of this difference, the proven equivalence of the nucleolus and the tau-value on the class of PMAS-admissible assignment games seems even more interesting.

\bibliographystyle{te}
\bibliography{Assignment-PMAS}
\end{document}